%% file: main.tex
\documentclass[twoside,leqno,letterpaper,11pt]{article}
\usepackage{amssymb,amsmath,amscd,latexsym,epic,eepic,gastex}
\usepackage{url,paralist}
\usepackage{luca,macros}
\usepackage{ltexpprt}
\usepackage{times}

\usepackage{xcolor}
\usepackage{colortbl}
\definecolor{mygray}{gray}{0.75}

\renewcommand{\dotsc}{{\ldots}}

\newcommand{\aut}{\mathcal{A}}
\newcommand{\ex}{\mathsf{ex}}

\newcommand{\M}{\mathsf{M}}
\renewcommand{\L}{\mathcal{L}}
\newcommand{\T}{\mathcal{T}}

\newcommand{\transg}{\trans^{G}}
\newcommand{\transgp}{\trans^{G'}}
\newcommand{\PM}{\mathit{PM}}

\providecommand{\tuple}[1]{\langle#1\rangle}
\providecommand{\abs}[1]{\lvert#1\rvert}
\providecommand{\set}[1]{\{#1\}}
\newcommand{\val}{\mathsf{val}}

\newcommand{\last}{\mathsf{last}}
\newcommand{\collapse}{\overline{\mathsf{obs}}}

\newtheorem{remark}{Remark}
\newtheorem{definition}{Definition}

\newcommand{\myomit}[1]{{}}

\begin{document}

\title{The Complexity of Synthesis from Probabilistic Components\thanks{
This research was supported by Austrian Science Fund (FWF) Grant No P23499- N23, FWF NFN Grant No S11407-N23 (RiSE), ERC Start grant (279307: Graph Games),
Microsoft Faculty Fellowship Award, 
NSF grants CNS 1049862 and CCF-1139011,
by NSF Expeditions in Computing project "ExCAPE: Expeditions in Computer 
Augmented Program Engineering", by BSF grant 9800096, and by gift from 
Intel.}
}
\author{
Krishnendu Chatterjee$^\dag$ \quad  Laurent Doyen$^{\S}$ \quad Moshe Y. Vardi$^{\ddag}$ \\ 
\normalsize
  $\strut^\dag$ IST Austria \\
\normalsize  $\strut^\S$ CNRS, LSV, ENS Cachan \\
\normalsize
  $\strut^\ddag$ Rice University, USA 
}

\date{}
\maketitle

\input{content}

\bibliographystyle{plain}
\bibliography{biblio}

\end{document}

%% file: content.tex
\input{intro}

\section{Definitions}

\smallskip\noindent{\bf Probability distributions.}
A probability distribution on a finite set $X$ is a function $f:
X\rightarrow [0,1]$ such that $\sum_{x\in X} f(x) = 1$. We use
$\dist(X)$ to denote the set of all probability distributions on set
$X$.

\subsection{Transducers} 
In this section we present the definitions of deterministic and 
probabilistic transducers, and strategies for them.

\smallskip\noindent{\bf Deterministic transducers.}
A \emph{deterministic transducer} is a tuple 
$B=  \tuple{ \Sigma_I, \Sigma_O, Q, q_0, \delta, L}$, where: 
$\Sigma_I$ is a finite input alphabet,
$\Sigma_O$ is a finite output alphabet,
$Q$ is a finite set of states,
$q_0\in Q$ is an initial state,
$L:Q\to\Sigma_O$ is an output function labeling states with output
letters, and $\delta:Q\times\Sigma_I \to Q$ is a transition
function. We define $\delta^\ast: \Sigma_I^\ast \rightarrow Q$ as
follows: 
$\delta^\ast(\epsilon) = q_0$ and for all $x\in \Sigma_I^\ast$ and 
$a\in \Sigma_I$, we have $\delta^\ast(x\cdot a) = \delta(\delta^\ast(x), a)$. 
%%We denote by $tree(B)$,  the $\Sigma_O$-labeled $\Sigma_I$-tree
%%$\tuple{\Sigma_I^\ast, \tau}$, where for all $x\in \Sigma_I^\ast$, 
%%we have $\tau(x) = L(\delta^\ast(x))$. We say $tree(B)$ is
%%the \emph{unwinding} of $B$. A $\Sigma$-labeled $D$-tree $T$ is called
%%\emph{regular}, if there exists a deterministic transducer $C$ such that $T = tree(C)$.  

\smallskip\noindent{\bf Probabilistic transducers.}
A {\em probabilistic transducer}, is a tuple
$\T = \tuple{ \Sigma_I, \Sigma_O, Q, q_0, \delta, F, L}$, where:
$\Sigma_I$ is a finite input alphabet,
$\Sigma_O$ is a finite output alphabet,
$Q$ is a finite set of states,
$q_0\in Q$ is an initial state,
$\delta:(Q\setminus F)\times\Sigma_I \to \dist(Q)$ is a probabilistic 
transition function,
$F\subseteq Q$ is a set of exit states, and
$L:Q\to\Sigma_O$ is an output function labeling states with output
letters. Note that there are no transitions out of an exit state. If $F$
is empty, we say $\T$ is a probabilistic transducer without exits. 
Note that deterministic transducers can be viewed as a special case of probabilistic 
transducers.

\smallskip\noindent{\bf Strategies for transducers.}
Given a probabilistic transducer 
$M = \tuple{\Sigma_I, \Sigma_O, Q, q_0, \delta, F, L}$, a
\emph{strategy} for $M$ is a function $f: Q^{+} \rightarrow
\dist(\Sigma_I)$ that probabilistically chooses an input for each finite 
sequence of states. 
We denote by $\calf$ the set of all strategies.
A strategy is memoryless if the
choice depends only on the last state in the sequence. A memoryless
strategy can be written as a function $g:Q\rightarrow\dist(\Sigma_I)$. 
A strategy is \emph{pure} if the choice is
deterministic. A pure strategy is a function $h: Q^{+} \rightarrow
\Sigma_I$, and a memoryless and pure strategy is a function $h: Q\rightarrow
\Sigma_I$. 

\smallskip\noindent{\bf Probability measure.} 
A strategy $f$ along with a probabilistic transducer $M$,
with set of states $Q$, induces a probability distribution on
$Q^\omega$, denoted $\mu_f$. By standard measure-theoretic arguments, it 
suffices to define $\mu_f$ for the cylinders of $Q^\omega$, which are sets of 
the form $\beta \cdot Q^\omega$,
where $\beta \in Q^\ast$. First we extend $\delta$ to exit states as
follows: for $a\in \Sigma_I$ and $q\in F$, $q' \in Q$, let
$\delta(q,a)(q) = 1$ and $\delta(q,a)(q') = 0$ if $q' \neq q$. 
Then we define $\mu_f(q_0\cdot Q^\omega) = 1$, and for  $\beta \in Q^\ast$, 
$q, q'\in Q$, we have $\mu_f(\beta qq'\cdot Q^\omega) = 
\mu_f(\beta q) \cdot \sum_{a\in \Sigma_I} (f(\beta q)(a)\cdot \delta(q,a)(q'))$.
These conditions say that there
is a unique start state, and the probability of visiting a state $q'$,
after visiting $\beta q$, is the same as the probability of the strategy 
picking a particular letter multiplied by the probability that  the transducer 
transitions from $q$ to $q'$ on that input letter, summed over all input 
letters.

\subsection{Library of Components}
A {\em library} is a set of probabilistic transducers 
that share the same input and output alphabets. Each transducer in the
library is called a \emph{component type}. Given a finite set of directions
$D$, we say a library $\L$ has width $D$, if each component type in the
library has exactly $|D|$ exit states. Since we can always add dummy
unreachable exit states to any component, we assume, w.l.o.g., that
all libraries have an associated width, usually denoted $D$. In the
context of a particular component type, we often refer to elements
of $D$ as exits, and subsets of $D$ as sets of exits. 
%Given a component $M$ from library $\L$, and
%a strategy $f$ for $M$, we say that exit $i\in D$ is
%\emph{selected} by $f$, if the $i$th exit state of $M$ is reachable
%under $f$.

\subsection{Controlflow Composition from Libraries}
We first informally describe the notion of controlflow
composition of components from a library as defined in~\cite{NLV14}.
The components in the
composition take turns interacting with the 
environment, and at each point in time, exactly one component is
active. When the active component reaches an exit state, control is
transferred to some other component. Thus, to define a controlflow
composition, it suffices to name the components used and describe how
control should be transferred between them. We use a deterministic
transducer to define the transfer of control.
Each library component can be used multiple times in a composition, and we
treat these occurrences as distinct \emph{component instances}. 
We emphasize that the composition can contain
potentially arbitrarily many instances of each component type inside it.
Thus, the size of the composition, a priori, is not bounded.
Note that our notion of composition is \emph{static}, where the components
called are determined before run time, rather than \emph{dynamic}, where 
the components called are determined during run time. 

Let $\L$ be a library of width $D$. A \emph{composer} over $\L$ is a
deterministic transducer 
$C = \tuple{D, \L, \mathcal{M}, \M_0, \Delta, \lambda}$. Here $\mathcal{M}$
is an arbitrary finite set of states. There is no bound on the size of
$\mathcal{M}$. Each $\M_i\in \mathcal{M}$ is %the name of an instance of 
a  component from $\L$ and $\lambda(\M_i) \in \L$ is the type of 
$\M_i$. We use the following notational convention for component
instances and names: the upright letter $\M$ always denotes
component names (i.e., states of a composer) and the italicized letter
$M$ always denotes the corresponding component instances (i.e., elements of
$\L$). Further, for notational convenience we often write $M_i$
directly instead of $\lambda(\M_i)$. Note that while each $\M_i$ is
distinct, the corresponding components $M_i$ need not 
be distinct. 
Each composer defines a unique composition over components from
$\L$. The current state of the composer corresponds to the component
that is in control. The transition function $\Delta$ describes how to transfer
control between components: $\Delta(\M,i) = \M'$ denotes that when the
composition is in the $i$th  final state of component $M$ it moves to
the start state of component $M'$. 
A composer can be viewed as an implicit
representation of a composition. We give an explicit definition of 
composition below.

\begin{definition}[Controlflow Composition]\label{def:composition}
Let $C = \tuple{D,\L,\mathcal{M}, \M_0,\Delta,
\lambda}$ be a composer over library $\L$ of width $D$, where
$\mathcal{M} = \{\M_0,\dotsc,\M_n\}$, $\lambda(\M_i) = \tuple{\Sigma_I, 
\Sigma_O, Q_i, q_0^i, \delta_i, F_i, L_i}$ and $F_i = \{q_x^i : x\in D\}$. 
The composition defined by $C$, denoted $\T_C$, is a probabilistic
transducer 
$\tuple{\Sigma_I,\Sigma_O, Q, q_0, \delta, \emptyset, L }$, where $Q
 =\bigcup_{i=0}^n (Q_i\times\{i\})$, $q_0 = \tuple{q^0_0,0}$, 
$L(\tuple{q,i}) = L_i(q)$, and the transition function $\delta$ is
defined as follows: For $\sigma\in \Sigma_I$, $\tuple{q,i}\in Q$ and  $\tuple{q',j} \in Q$,

\vspace{2pt}
\begin{enumerate}
\item If $q\in Q_i\setminus F_i$, then
\[
\qquad \delta(\tuple{q,i},\sigma)(\tuple{q',j}) =
\begin{cases}
\delta_i(q,\sigma)(q') &  \text{if $i = j$} \\
0 & \text{otherwise}
\end{cases}
\]
\item If $q = q_x^i \in F_i$, where  $\Delta(\M_i,x) = \M_k$, then 
\[
 \qquad \delta(\tuple{q,i},\sigma)(\tuple{q', j}) = 
\begin{cases}
1 & \text{if $j$ = $k$ and $q' = q_0^k$} \\
0 & \text{otherwise}
\end{cases}
\]
\end{enumerate}
\end{definition}

Note that the composition is a probabilistic transducer without exits.
When the composition is in a state $\tuple{q,i}$ corresponding to a non-exit
state $q$ of component $M_i$, it behaves like $M_i$. When the
composition is in a state $\tuple{q_f, i}$ corresponding to an exit state
$q_f$ of component $M_i$, the control is transferred to the start state
of another component as determined by the transition function of the
composer. Thus, at each point in time, only one component is active
and interacting with the environment.

\subsection{Parity objectives and values for probabilistic transducer}
An \emph{index function} for a transducer is a function that assigns
a natural number, called a priority index, to each state of the transducer.
An index function $\alpha$ defines a parity objective $\Phi_{\alpha}$ 
that is the subset of $Q^\omega$
that consists of the set of infinite sequence of states such that 
the minimum priority that is visited infinitely often is even.
Given a probabilistic transducer $\T$ and a parity objective $\Phi$, the 
value of the probabilistic transducer for the objective, denoted as 
$\val(\T,\Phi)$, is $\inf_{f \in \calf} \mu_f(\Phi)$.
In other words, it is the minimal probability with which the parity objective 
is satisfied over all strategies in the transducer.

%An index function for a library is a function that assigns a priority 
%to every state of every component in the library. Given an index
%function $\alpha$ for a library $\L$, we define $\max(\alpha)$ to be
%the highest priority assigned by $\alpha$. We can assume, w.l.o.g., that
%$\max(\alpha)$ is not larger than  twice the maximal number of
%states in the components of the library. 
%Given a transducer $M$, index function $\alpha$, and a strategy $f$
%for $M$, we say $f$ \emph{visits} priority $p$ if there
%exists a state $q$ of $M$ such that $\alpha(q) = p$ and $q$ is
%reachable under $f$.

\subsection{The synthesis questions}
In this work we consider two types of synthesis questions for controlflow
composition. 
In the first problem (namely, synthesis for embedded parity) the parity 
objective is specified directly on the state space of the library components, 
and in the second problem (namely, synthesis from DPW specifications) the 
parity objective is specified by a separate deterministic parity automaton.

\subsubsection{Synthesis for Embedded Parity}\label{sec:embedded-parity}
We first consider an index function that associates to each state of the components
in the library a priority, and a specification defined as a parity condition over 
the sequence of visited states.

%Let $M$ be a probabilistic transducer and $\alpha$ be an index
%function. A strategy $f$ for $M$ is \emph{winning} for the
%environment if with positive probability the highest priority visited
%infinitely often (i.o.) is odd.
%We say that $M$
%\emph{satisfies} $\alpha$ if there exists no winning strategy for the
%environment. Given a composer $C$ over library $\L$, we say that $C$
%\emph{satisfies} $\alpha$ if $\T_C$ satisfies $\alpha$.

\smallskip\noindent{\em Exit control relation.}
Given a library $\L$ of width $D$, an \emph{exit control relation} is a
set $R\subseteq D \times \L$. 
We say that a composer $C = \tuple{D, \L, \mathcal{M}, \M_0, \Delta,
\lambda}$ is \emph{compatible}
with $R$, if the following holds: for all $\M, \M'\in \mathcal{M}$ and
$i\in D$, if $\Delta(\M, i) = \M'$ then $\tuple{i,M'} \in R$. Thus, each
element of $R$ can be viewed as a constraint on how the composer is
allowed to connect components.
An exit control relation is \emph{non-blocking} if for every $i \in D$ there 
exists a component $M \in \L$ such that $\tuple{i,M} \in R$ (i.e., every 
exit has at least one possible component for the next choice).
For technical convenience we only consider non-blocking exit control
relations.
If $R=D \times \L$ (i.e., there is no constraint on the composer to
connect components), then we refer to the relation as \emph{unrestricted}
exit control relation.

\begin{definition}[Embedded parity realizability and synthesis.]
Consider a library $\L$ of width $D$, an exit control relation $R$ for $\L$, 
and an index function $\alpha$ for the components in $\L$ that defines the 
parity objective $\Phi_{\alpha}$. 
The \emph{qualitative} (resp., \emph{quantitative}) 
\emph{embedded parity realizability problem} is to decide whether there exists 
a composer $C$ over $\L$, such that $C$ is compatible with $R$, and 
$\val(\T_C,\Phi_{\alpha})=1$ (resp., $\val(\T_C,\Phi_{\alpha}) \geq \eta$, 
for a given rational threshold $\eta \in (0,1)$).
A witness composer for the qualitative problem is called an almost-sure 
composer, and for the quantitative problem is called an $\eta$-optimal 
composer.
%%If such a composer exists, we say that $\L$ \emph{realizes} $\alpha$ under $R$.
The corresponding \emph{embedded parity synthesis problems} are to find such a 
composer $C$ if it exists.
\end{definition}

\subsubsection{Synthesis for DPW Specifications}
A \emph{deterministic parity automaton (DPW)} is a deterministic transducer 
where the labeling function is an index function that defines a parity objective. 
Given a DPW $A$, every word (infinite sequence of input letters) induces a run 
of the automaton, which is an infinite sequence of states, and the word is
accepted if the run satisfies the parity objective. 
The language $L_A$ of a DPW $A$ is the set of words accepted by $A$.
Let $A$ be a deterministic parity automaton (DPW), 
$M$ be a probabilistic transducer and $\L$ be a library of components. 
We say $A$ is a \emph{monitor} for $M$ (resp. $\L$) if the input alphabet 
of $A$ is the same as the output alphabet of $M$ (resp. $\L$).
Let $A$ be a monitor for $M$ and let $L_A$ be the language accepted by $A$.
The value of $M$ for $A$, denoted as $\val(M,A)$, is 
$\inf_{f \in \calf} \mu_f(\lambda^{-1}(L_A))$.
Note that the compatibility of the composer with an exit control relation 
can be encoded in the DPW (without loss of generality, we do not allow two 
distinct exit states to have the same output).

\begin{definition}[DPW realizability and synthesis.]
Consider  a library $\L$ and a DPW $A$ that is a monitor for $\L$. 
The \emph{qualitative} (resp., \emph{quantitative}) 
\emph{DPW probabilistic realizability problem} is to decide
whether there exists a  composer $C$ over $\L$, such that 
$\val(\T_C,A)=1$ (resp., $\val(\T_C,A) \geq \eta$, 
for a given rational threshold $\eta \in (0,1)$).
%%If such a composer exists, we say that $\L$ \emph{realizes} $A$. 
A witness composer for the qualitative problem is called an almost-sure 
composer, and for the quantitative problem is called an $\eta$-optimal 
composer.
The corresponding \emph{DPW probabilistic synthesis problems} are to
find such a composer $C$ if it exists.
\end{definition}

\begin{remark}\label{rem:partial}
We remark that the realizability problem for libraries with components 
can be viewed as a 2-player partial-observation stochastic parity game. 
Informally, the game can be described as follows: 
the two players are the composer $C$ and the environment $E$. 
The $C$ player chooses components and the $E$ player chooses sequence of 
inputs in the components chosen by $C$. However, $C$ cannot see the inputs
of $E$ or even the length of the time inside a component. 
At the start $C$ chooses a component $M$ from the library $\L$. 
The turn passes to $E$, who chooses a sequence of inputs, inducing a 
probability distribution over paths in $M$ from its start state to some exit 
$x$ in $D$. 
The turn then passes to $C$, which must choose some component $M'$ in $\L$ 
and pass the turn to $E$ and so on. 
As $C$ cannot see the moves made by $E$ inside $M$, the choice of $C$ 
cannot be based on the run in $M$, but only on the exit induced by the inputs 
selected by $E$ and previous moves made by $C$. So $C$ must choose the same next
component $M'$ for different runs that reach exit $x$ of $M$. In general,
different runs will visit different priorities inside $M$. This is a
two-player stochastic parity game where one of the players (namely the 
composer~$C$) does not have full information. However, there is also a crucial 
difference from traditional partial-observation games, as the composer does 
not even see the number of steps executed inside a component. 
If $C$ has a winning strategy that requires finite memory, then such a 
strategy would yield a suitable finite composer to satisfy the parity 
objective defined by the index function $\alpha$, thus solving the synthesis 
problem.
\end{remark}

\section{The Complexity of Realizability for Embedded Parity}
In this section we will establish the results for the complexity of 
realizability for embedded parity. 
In~\cite{NLV14} the problem was interpreted as a partial-observation game
(see Remark~\ref{rem:partial}). 
While the natural interpretation of the  embedded parity problem is a 
partial-observation game, we show how the problem can be interpreted as a 
perfect-information stochastic game.

\subsection{Perfect-information Stochastic Parity Games}
In this section we present the basic definitions and results
for perfect-information stochastic games.

\noindent{\bf Perfect-information stochastic games.}
A perfect-information stochastic game consists of a tuple 
$G=\tuple{S,S_1,S_2,A_1,A_2,\transg}$, where $S$ is a finite set of states 
partitioned into player-1 states (namely, $S_1$) and player-2 states 
(namely $S_2$), $A_1$ (resp., $A_2$) is the set of actions for player~1 
(resp., player~2), and $\transg: (S_1 \times A_1) \cup (S_2 \times A_2) 
\to \dist(S)$ is a probabilistic transition function that given a player-1 
state and player-1 action, or a player-2 state and a player-2 action gives
a probability distribution over the successor states.
If the transition function is \emph{deterministic} (that is the codomain
of $\transg$ is $S$ instead of $\dist(S)$), then the game is a perfect-information deterministic game.

\smallskip\noindent{\bf Plays and strategies.} 
A \emph{play} is an infinite sequence of state-action pairs 
$\tuple{s_0 a_0 s_1 a_1 \ldots}$ such that for all $j \geq 0$ we have that 
if $s_j \in S_i$ for $i\in \set{1,2}$, then $a_j \in A_i$ and 
$\transg(s_j,a_j)(s_{j+i} )>0$.
A strategy is a recipe for a player to choose actions to extend finite prefixes of plays.
Formally, a strategy $\straa$ for player~1 is a function 
$\straa: S^\ast \cdot S_1 \to \dist(A_1)$ that given a 
finite sequence of visited states gives a probability distribution over the actions 
(to be chosen next).
A \emph{pure} strategy chooses a deterministic action, i.e., is a function
$\straa: S^\ast \cdot S_1 \to A_1$.
A pure memoryless strategy is a pure strategy that does not depend on the 
finite prefix of the play but only on the current state, i.e., is a function 
$\straa: S_1 \to A_1$.
The definitions for player-2 strategies $\strab$ are analogous.
We denote by $\Straa$ (resp., $\Straa^{\PM}$) the set of all (resp., all pure 
memoryless) strategies for player~1, and analogously $\Strab$ 
(resp., $\Strab^{\PM}$ for player~2).
Given strategies $\straa \in \Straa$ and $\strab \in \Strab$, and a starting 
state $s$, there is a unique probability measure over events (i.e., measurable subsets
of $S^{\omega}$), which is denoted as $\Prb_s^{\straa,\strab}(\cdot)$.

\smallskip\noindent{\em Finite-memory strategies.}
A pure player-1 strategy uses \emph{finite-memory} if it can be encoded
by a transducer $\tuple{\mem, m_0, \straa_u, \straa_n}$
where $\mem$ is a finite set (the memory of the strategy), 
$m_0 \in \mem$ is the initial memory value,
$\straa_u: \mem \times S \to \mem$ is the memory-update function, and 
$\straa_n: \mem \to A_1$ is the next-action function. 
Note that a finite-memory strategy is a deterministic transducer with 
input alphabet $S$, output alphabet $A_1$, where $\straa_u$ is the deterministic
transition function, and $\straa_n$ is the output labeling function.
However, for finite-memory strategies, since the input and output is always
the set of states and actions for player~1, for simplicity, 
we will represent them as a tuple  $\tuple{\mem, m_0, \straa_u, \straa_n}$.
The \emph{size} of the strategy is the number $\abs{\mem}$ of memory values.
If the current state is $s$, and the current memory value is $m$,
then the memory is updated to $m' = \straa_u(m,s)$, and 
the strategy chooses the next action $\straa_n(m')$.
Formally, $\tuple{\mem, m_0, \straa_u, \straa_n}$
defines the strategy $\straa$ such that $\straa(\rho)= 
\straa_n(\widehat{\straa}_u(m_0, \rho))$ for all $\rho \in S^+$, where 
$\widehat{\straa}_u$ extends $\straa_u$ to sequences of states  as expected. 
%\mynote{K to L: Pls check above is consistent with proof notation.}
%\mynote{L: updated intuitive explanation.}

%%\mynote{K: Above is fixed, slightly differently. Previous version with comment below.}
%%$\straa(\rho\cdot s)=\straa_n(\widehat{\straa}_u(m_0, \rho\cdot s))$
%%\mynote{L: change to $\straa(s \cdot \rho) = \straa_n(\widehat{\straa}_u(m_0, \rho))$
%%for all $s \in S$ and $\rho \in S^*$,}
%%for all $\rho \in S^*$ and $s \in S_1$, where $\widehat{\straa}_u$ extends 
%%$\straa_u$ to sequences of states  as expected. 
%This definition extends to infinite-memory strategies by not restricting $\mem$ to be finite.

\smallskip\noindent{\bf Parity objectives, almost-sure, and value problem.} 
Given a perfect-information stochastic game, a parity objective is defined by an index
function $\alpha$ on the state space.
Given a strategy $\straa$, the value of the strategy in a state $s$ of the game $G$
with parity objective $\Phi_{\alpha}$, denoted by $\val^G(\straa,\Phi_{\alpha})(s)$, 
is the infimum of the probabilities among all 
player-2 strategies, i.e., $\val^G(\straa,\Phi_{\alpha})(s)= 
\inf_{\strab \in \Strab} \Prb_s^{\straa,\strab}(\Phi_{\alpha})$.
The value of the game is $\val^G(\Phi_{\alpha})(s) 
=\sup_{\straa \in \Straa} \val^G(\straa,\Phi_{\alpha})(s)$.
A strategy $\straa$ is almost-sure winning from $s$ 
if $\val^G(\straa,\Phi_{\alpha})(s)=1$.
The following theorem summarizes the basic results about 
perfect-information games.

\begin{theorem}\label{thm:perfect} 
The following assertions hold~{\rm \cite{CJH04,ChaThesis,CF11,CH08,AM09}}:
\begin{compactenum}
\item \emph{(Complexity).} 
The quantitative decision problem (of whether $\val^G(\Phi_{\alpha}) \geq \eta$, 
given rational $\eta \in (0,1]$) for perfect-information stochastic parity games
lies in NP $\cap$ coNP (also UP $\cap$ coUP).

\item \emph{(Memoryless determinacy).} We have 
\[
\val^G(\Phi_{\alpha})(s) = 
\sup_{\straa \in \Straa^{\PM}} \inf_{\strab \in \Strab} 
\Prb_s^{\straa,\strab}(\Phi_{\alpha})
=
\inf_{\strab \in \Strab^{\PM}} \sup_{\straa \in \Straa} 
\Prb_s^{\straa,\strab}(\Phi_{\alpha}),
\]
i.e., the quantification over the strategies can be restricted 
to $\straa \in \Straa^{PM}$ and $\strab \in \Strab^{PM}$.
\end{compactenum}
\end{theorem}

%{\bf TO BE DEFINED. PARITY OBJECTIVES, ALMOST-SURE AND VALUE PROBLEM. MEMORYLESS SUFFICIENCY. COMPLEXITY.}

\subsection{Complexity Results}\label{sec:complexity-results}
We first present the reduction for upper bounds.

\smallskip\noindent{\bf The upper-bound reduction.}
Consider a library $\L$ of width $D$, a non-blocking exit control relation $R$ for $\L$, 
and an index function $\alpha$ for $\L$ that defines the parity 
objective $\Phi_{\alpha}$.
Let the number of components be $k+1$, and let $M_i=\tuple{\Sigma_I, \Sigma_O,
Q_i, q_0^i, \delta_i, F_i, L_i}$ for $0\leq i \leq k$.
Let us denote by $[k]=\set{0,1,2,\ldots,k}$.
We define a perfect-information stochastic game $G_{\L}=\tuple{S,S_1,S_2,A_1,A_2,\transg_{\L}}$ 
with an index function $\alpha_G$ as follows:
$S=\bigcup_{i=0}^k (Q_i \times \set{i}) \cup \set{\bot}$, $S_1= \bigcup_{i=0}^k (F_i \times \set{i})$, $S_2=S\setminus S_1$,
$A_1=[k]$, and $A_2=\Sigma_I$. 
The state $\bot$ is a losing absorbing state (i.e., a state with self-loop as the only 
outgoing transition and assigned odd priority by the index function $\alpha_G$), 
and the other transitions defined by the function $\transg_{\L}$ are as follows:
(i)~for $s=\tuple{q,i} \in S_2$, and $\sigma \in A_2$ 
\[
\qquad \transg_{\L}(\tuple{q,i},\sigma)(\tuple{q',j}) =
\begin{cases}
\delta_i(q,\sigma)(q') &  \text{if $i = j$} \\
0 & \text{otherwise}
\end{cases}
\]
(ii)~for $s=\tuple{q^i_x,i} \in S_1$ and $j \in [k]$, we have that
if $\tuple{x,M_j} \in R$,  then $\transg_{\L}(\tuple{q^i_x,i},j)(\tuple{q_0^j,j})=1$, 
else $\transg_{\L}(\tuple{q^i_x,i},j)(\bot)=1$.
The intuitive description of the transitions is as follows: 
(1)~Given a player-2 state
that is a non-exit state $q$ in a component $M_i$, and an action for player~2 that is an input
letter, the transition function $\transg_{\L}$ mimics the transition $\delta_i$ of
$M_i$; and 
(2)~given a player-1 state that is an exit state $q^{i}_x$ in component $i$, and 
an action for player~1 that is the choice of a component $j$, if $\tuple{x,M_j}$ is allowed
by $R$, then the next state is the starting state of component $j$, and if the choice
$\tuple{x,M_j}$ is invalid (not allowed by $R$), then the next state is the losing absorbing
state $\bot$.
For all $\tuple{q,i} \in S \setminus \set{\bot}$ we have $\alpha_G(\tuple{q,i})=\alpha(q)$,
and we denote by $\Phi_{\alpha_G}$ the parity objective in $G_{\L}$.
Note the similarity of the state space description in comparison with 
Definition~\ref{def:composition} for controlflow composition.

\smallskip\noindent{\bf Correctness of reduction.} 
There are two steps to establish correctness of the reduction.
The first step is given a composer for $\L$ to construct a finite-memory strategy 
for player~1  in $G_{\L}$.
Intuitively, this is simple as a composer represents a strategy for a 
partial-observation game (Remark~\ref{rem:partial}), whereas in $G_{\L}$ we 
have perfect information.
However, not every strategy in $G_{\L}$ can be converted to a composer 
(i.e., a perfect-information game strategy cannot be converted to a partial-observation
strategy).
But we show that a pure memoryless strategy in $G_{\L}$ can be converted to a 
composer. 
While~\cite{NLV14} treats the problem as a partial-observation game, our insight to convert
pure memoryless strategies of the perfect-information game $G_{\L}$ to composers 
allows us to establish the correctness with respect to $G_{\L}$.
We present both the steps below.

\begin{lemma}
Consider a library $\L$ of width $D$, a non-blocking exit control relation $R$ for $\L$, 
and an index function $\alpha$ for $\L$ that defines the parity 
objective $\Phi_{\alpha}$. 
Let $G_{\L}$ be the corresponding perfect-information stochastic game with 
parity objective $\Phi_{\alpha_G}$.
For all composers $C$, and the corresponding strategy $\straa_C$ in $G_{\L}$ we 
have $\val(\T_C,\Phi_{\alpha})=\val^{G_{\L}}(\straa_C,\Phi_{\alpha_G})(\tuple{q_0^0,0})$.
\end{lemma}

\begin{proof}
{\em Composer to finite-memory strategies.}
Given a composer $C = \tuple{D,\L,\mathcal{M}, \M_0,\Delta,\lambda}$ 
for the library we define a finite-memory strategy 
$\straa_C=\tuple{\mem,m_0,\straa_u,\straa_n}$ for the perfect-information 
stochastic game $G_{\L}$ as follows:
(a)~$\mem=\mathcal{M}$ and $m_0=\M_0$; and
(b)~$\straa_u(\M_i,s)=\M_i$ for $s\in S_2$, and 
$\straa_u(\M_i,\tuple{q^j_x,j})=\M_\ell$ for $s\in S_1$ 
if $\Delta(\M_i,x)=\M_\ell$, where $\lambda(\M_i)=M_j$; 
and (c)~$\straa_n(\M_i)=j$ where $\lambda(\M_i)=M_j$. 

%%\mynote{K: Above is fixed. Pls check. Previous version and comment below.} 
%%$\straa_u(\M_i,s)=\M_i$ for $s\in S_2$ and 
%%$\straa_u(\M_i,\tuple{q^i_x,j})=\M_\ell$ \mynote{(1)} for $s\in S_1$ if $\Delta(\M_i,x)=\M_\ell$; 
%%and $\straa_n(\M_i)=j$ where $\lambda(\M_i)=M_j$ \mynote{(2)}. 

%%\mynote{L: in (1) can it be that $i \neq j$, why is the exit state named $q^i_x$
%%as there is no reason that $\M_i$ is of type $M_i$; in (2) the type of $\M_i$ is $M_j$.}

In other words, the finite-memory strategy has the same state space as the 
composer, and if the current state is a player-2 state, then it does not update
the memory state, and given the current state is a player-1 state it updates 
it memory state according to the transition function of the composer, and the action 
played is according to the labeling function of the transducer.
In other words, the strategy $\straa_C$ mimics the composer, and there is a 
one-to-one correspondence between strategies for player~2 in the 
perfect-information stochastic game, and the strategies of the environment in 
$\T_C$. 
Without loss of generality we consider that the composer must always start with the first component
(i.e., $M_0$) and hence the starting state is $\tuple{q_0^0,0}$.
%%\mynote{K: Added the above sentence.}
This gives us the desired result.
\hfill\qed
\end{proof}

For the second step we first consider valid pure memoryless strategies.

\smallskip\noindent{\em Valid pure memoryless strategies in $G_{\L}$.}
A pure memoryless strategy $\straa$ in $G_{\L}$ is \emph{valid} if the following 
condition holds: for all states $\tuple{q_x^{i},i} \in S_1$ if $\straa(\tuple{q_x^{i},i})=j$,
then $\tuple{x,M_j} \in R$, i.e., the choices of the pure memoryless strategies respect
the exit control relation.

\begin{lemma}
Consider a library $\L$ of width $D$, a non-blocking exit control relation $R$ for $\L$, 
and an index function $\alpha$ for $\L$ that defines the parity 
objective $\Phi_{\alpha}$. 
Let $G_{\L}$ be the corresponding perfect-information stochastic game with 
parity objective $\Phi_{\alpha_G}$.
For all valid pure memoryless strategies $\straa$ in $G_{\L}$, and 
the corresponding composer  $C_{\straa}$, we have 
$\val(\T_{C_{\straa}},\Phi_{\alpha})=\val^{G_{\L}}(\straa,\Phi_{\alpha_G})(\tuple{q_0^0,0})$.
\end{lemma}

\begin{proof}
{\em Valid pure memoryless strategies to composers.}
Given a valid pure memoryless strategy $\straa$ in $G_{\L}$ we define a composer  
$C_{\straa}=\tuple{D,\L,\mathcal{M}, \M_0,\Delta,\lambda}$ as follows:
$\mathcal{M}=[k]$, $\M_0=0$, $\lambda(i)=M_i$, and 
for $0\leq i \leq k$ and $x \in D$ we have that $\Delta(i,x)=j$ where
$\straa(\tuple{q^i_x,i})=j$ for $q^{i}_x \in F_i$.
In other words, for the composer there is a state for every component, and given a 
component and an exit state, the composer plays as the pure memoryless strategy. 
Note that since $\straa$ is valid, the composer obtained from $\straa$ 
is compatible with the relation $R$. 
Note that the composer mimics the pure memoryless strategy, and there is a 
one-to-one correspondence between strategies of player~2 in $G_{\L}$ and 
strategies of the environment in $\T_{C_\straa}$, which establishes the result.
\hfill\qed
\end{proof}

\begin{lemma}\label{lemma:key_1}
Consider a library $\L$ of width $D$, a non-blocking exit control relation $R$ for $\L$, 
and an index function $\alpha$ for $\L$ that defines the parity 
objective $\Phi_{\alpha}$. 
Let $G_{\L}$ be the corresponding perfect-information stochastic game with 
parity objective $\Phi_{\alpha_G}$.
There exists an almost-sure composer iff there exists an almost-sure winning strategy
in $G_{\L}$ from $\tuple{q_0^0,0}$, and there exists an $\eta$-optimal composer iff the value in 
$G_{\L}$ at $\tuple{q_0^0,0}$ is at least $\eta$.
\end{lemma}

\begin{proof}
{\em Sufficiency of valid pure memoryless strategies.}
Note that in the game $G_{\L}$ we can restrict to only valid pure memoryless strategies 
because for a state $s \in S_1$ if an action is chosen that is not allowed by a 
valid strategy, then it leads to the losing absorbing state $\bot$. 
Thus in $G_{\L}$ if there is an almost-sure winning strategy (resp., a strategy 
to ensure that the value is at least $\eta$), then there is a valid pure memoryless
strategy to ensure the same (Theorem~\ref{thm:perfect}).
\hfill\qed
\end{proof}

Lemma~\ref{lemma:key_1} along with Theorem~\ref{thm:perfect}
%%the existing complexity results for perfect-information stochastic parity games 
imply that the qualitative and quantitative problems for the 
realizability for embedded parity lie in NP $\cap$ coNP (also UP $\cap$ coUP). 
We now present two related results.
We first show that with unrestricted exit control relation the problem
can be solved in polynomial time, and then show that in general (i.e., with exit 
control relation) the problem is at least as hard as solving perfect-information
deterministic parity games (for which no polynomial-time algorithm is known).

\smallskip\noindent{\bf The unrestricted exit control relation problem.}
For the unrestricted exit control relation problem, we modify the construction 
of the game $G_{\L}$ as follows: we add another state $\top$ that belongs to player~1,
and for every state $s \in S_1$ and every action $a \in A_1$ the next state is $\top$,
and from $\top$ player~1 can choose any component (i.e., state $\tuple{q_0^j,j}$ in $G_{\L}$).
We refer to the modified game as $\ov{G}_{\L}$. 
If the exit control relation is unrestricted then the result corresponding to 
Lemma~\ref{lemma:key_1} holds for $\ov{G}_{\L}$. 
However, in $\ov{G}_{\L}$ the number of memoryless player-1 strategies is only $k+1$
(one each for the choice of each component at $\top$), and once a 
memoryless strategy for player~1 is fixed we obtain an MDP which can be solved 
in polynomial time (and the qualitative problem in strongly polynomial
time by discrete graph theoretic algorithms)~\cite{CY95,CJH03,CH11,CH14}.
Hence it follows that both the qualitative and quantitative realizability and 
synthesis problems for embedded parity with unrestricted exit control relation 
can be solved in polynomial time.

\smallskip\noindent{\bf The hardness reduction.}
We now present a reduction form perfect-information deterministic parity games 
to the realizability problem for embedded parity. 
For simplicity we consider \emph{alternating} games where the players make 
move in alternate turns, i.e., we consider a perfect-information game $G=\tuple{S,S_1,S_2,A_1,A_2,\transg}$
where $\transg$ is deterministic and is decomposed into two functions 
$\transg_1: S_1 \times A_1 \to S_2$ and $\transg_2: S_2 \times A_2 \to S_1$
(player-1 move leads to player-2 state and vice versa).
A perfect-information deterministic game with an index function $\alpha$ can be converted
to an equivalent alternating game with a linear blow-up by adding dummy states. 
Given an alternating perfect-information game, let 
$S_2=\set{s_0^2,s_1^2,s_2^2,\ldots,s_k^2}$ and 
$S_1=\set{s_1^1,s_2^1,\ldots,s_d^1}$. 
We construct $\L$ of width $S_1$, an exit control relation $R$, and an index function 
$\alpha'$ as follows: 
(a)~there are $k+1$ components $M_0, M_1, \ldots, M_k$ 
one for each state in $S_2$; 
(b)~each $M_i=\tuple{ \Sigma_I, \Sigma_O, Q_i, q_0^{i}, \delta_i, F_i, L_i}$ is defined as follows:
(i)~$\Sigma_I=A_2$, (ii)~$\Sigma_O$ and $L_i$ are not relevant for the reduction, 
(iii)~$Q_i=\set{q_0^{i}, q_1^{i},\ldots, q_d^{i}}$ with $F_i=Q_i \setminus \set{q_0^{i}}$, and 
(iv)~$\delta_i(q_0^{i},\sigma)=q_j^{i}$ where $s_j^1=\transg_2(s_i^2,\sigma)$ for all $\sigma \in \Sigma_I$.
The exit control relation $R$ is as follows: $R=\set{\tuple{i,M_j} \mid \exists a_1 \in A_1. \ \transg_1(s_i^1,a_1) =s_j^2}$.
Intuitively, there exists a component for each state in $S_2$, and the exit states for each component
corresponds to states of player~1.
In the start state $q_0^{i}$ for component $M_i$ the transition function $\delta_i$ represents the
transition function $\transg_2$ of the perfect-information game, i.e., given an input letter $\sigma$
which is an action for player~2, if the transition given $\sigma$ is from $s^2_i$ to $s^1_j$, then 
in $M_i$ there is a corresponding transition to $q^{i}_j$.
The exit control relation represents the transitions for player~1. 
Since the exit states in each component is reached in one step, the composer strategies for the
library represents perfect-information strategies.
The index function $\alpha'$ is as follows: $\alpha'(q_0^{i})=\alpha(s_i^2)$ and 
$\alpha'(q_j^{i})=\alpha(s_j^1)$ for $j \geq 1$.
There exists an almost-sure winning strategy in $G$ from $s_0^2$ iff there exists an almost-sure
composer for $\L$ with $R$.
Also note that for the reduction every component transducer is deterministic.

\begin{theorem}[Complexity of embedded parity realizability]
Consider a library $\L$ of width $D$, a non-blocking exit control relation $R$ for $\L$, 
and an index function $\alpha$ for $\L$ that defines the parity 
objective $\Phi_{\alpha}$. The following assertions hold:
\begin{compactenum}
\item The realizability problem belongs to NP $\cap$ coNP (also UP $\cap$ coUP),
and is at least as hard as the (almost-sure) decision problem for 
perfect-information deterministic parity games.
\item If $R$ is an unrestricted exit control relation, then the realizability and 
synthesis problems can be solved in polynomial time
\end{compactenum}
\end{theorem}

\section{The Complexity of Realizability for DPW Specifications}
In this section we present three results. 
First, we present a new result for partial-observation stochastic 
parity games. 
Second, we show that the qualitative realizability problem for 
DPW specifications can be reduced to our solution for partial-observation
stochastic games yielding an EXPTIME-complete result for the problem.
Finally, we show that the quantitative realizability problem for 
DPW specifications is undecidable.  
%\mynote{K: Added paragraph.}

\subsection{Partial-observation Stochastic Parity Games\label{sec:posp-games}} In this section we  
consider partial-observation games with various restrictions on strategies and present 
a new result for a class of strategies that correspond to the analysis of the 
qualitative realizability problem.
%\mynote{K: Added subsection and paragraph.}

\smallskip\noindent{\bf Partial-observation stochastic games.}
In a stochastic game with partial observation, some states are not distinguishable
for player~$1$. We say that they have the same observation for player~$1$.
Formally, a partial-observation stochastic game consists of a stochastic game 
$G = \tuple{S,S_1,S_2,A_1,A_2,\transg}$, a finite set $\Obs$ of observations,
and a mapping $\obs: S \to \Obs$ that assigns to each state $s$ of the game
an observation $\obs(s)$ for player~$1$.

\smallskip\noindent{\bf Observational equivalence and strategies.}
The observation mapping induces indistinguishability of play prefixes for 
player~$1$, and therefore we need to consider only the player-$1$ strategies 
that play in the same way after two indistinguishable play prefixes. 
We consider three different classes of strategies depending on the 
indistinguishability of play prefixes for player~1 and they are as follows:
(i)~the play prefixes have the same observation sequence; 
(ii)~the play prefixes have the same observation sequence, except that the 
last observation may be repeated arbitrarily many times; and
(iii)~the play prefixes have the same sequence of distinct observations, 
that is they have the same observation sequence up to repetition (stuttering).
We now formally define the classes of strategies.
%\mynote{K: Modified paragraph slightly.}

\smallskip\noindent{\bf Classes of strategies.}
The \emph{observation sequence} of a sequence $\rho=s_0 s_1 \dots s_n$
is the sequence $\obs(\rho) = \obs(s_0) \dots \obs(s_n)$ of state observations;
the \emph{collapsed stuttering} of $\rho$ is the sequence 
$\collapse(\rho) = o_0 o_1 o_2 \ldots$ of distinct observations
defined inductively as follows:
$o_0 = \obs(s_0)$ and for all $i \geq 1$ we have $o_i = \obs(s_i)$ 
if $\obs(s_i) \neq \obs(s_{i-1})$, and $o_i = \epsilon$ otherwise
(where $\epsilon$ is the empty sequence). 
We consider three types of strategies. 
A strategy $\straa$ for player~$1$~is 

\begin{itemize}

\item \emph{observation-based} if
for all sequences $\rho, \rho' \in S^+$ such that $\last(\rho) \in S_1$ and 
$\last(\rho') \in S_1$, if $\obs(\rho) = \obs(\rho')$ then $\straa(\rho) = \straa(\rho')$;

\item \emph{observation-based stutter-invariant} if it is observation-based 
and for sequences $\rho \in S^+$, % such that $\last(\rho) \in S_1$, 
for all states $s \in S$, if $\obs(s) = \obs(\last(\rho))$, then $\straa(\rho) = \straa(\rho s)$;
%\mynote{K: New vs old definition: pls check.}
%\mynote{L: new definition updated.}

%\item \emph{observation-based stutter-invariant} if it is observation-based and
%for all play prefixes $\rho, \rho'$ such that $\last(\rho) \in S_1$ and 
%$\last(\rho') \in S_1$, if $\obs(\rho) = \obs(\rho')\cdot o$ where 
%$o = \last(\obs(\rho'))$, then $\straa(\rho) = \straa(\rho')$;

\item \emph{observation-based collapsed-stutter-invariant} 
if for all sequences $\rho, \rho' \in S^+$ such that $\last(\rho) \in S_1$ and 
$\last(\rho') \in S_1$, if $\collapse(\rho) = \collapse(\rho')$,
then $\straa(\rho) = \straa(\rho')$.

\end{itemize}

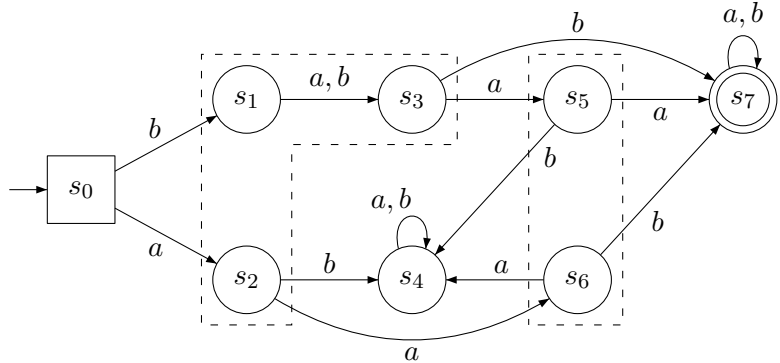
\begin{figure}[!tb]
  \begin{center}
    \hrule
    \input{figures/stuttering-vs-collapsed-stuttering.tex}
    \hrule
    \caption{Reachability game where there is no collapsed-stuttering invariant winning strategy,
but there is a observation-based stutter-invariant winning strategy. States $\{s_1,s_2,s_3\}$ are
indistinguishable for player~$1$, as well as $\{s_5,s_6\}$. The objective for player~$1$
is to reach $s_7$.\label{fig:stuttering}}
  \end{center}
\end{figure}

%\begin{figure}[!tb]
%  \begin{center}
%    \hrule
%    \input{figures/components.tex}
%    \hrule
%    \caption{Component library and DPW specification.\label{fig:components}}
%  \end{center}
%\end{figure}

The key difference between observation-based stutter-invariant and observation-based
collapsed-stutter-invariant strategies is as follows: they both must play the 
same action while an observation is repeated, however, collapsed-stutter-invariant 
strategies cannot observe the length of the repetitions of an observation,
whereas a stutter-invariant strategy can. 
In the following example we illustrate this difference.
%\mynote{K: Added paragraph. Example to be done. Below changed Larsen to Cassez.}
%\mynote{L: Modified above paragraph slightly. Added example below.}

\smallskip\noindent{\bf Example.}
\figurename~\ref{fig:stuttering} shows a (non-stochastic) reachability game
where the objective of player~$1$ is to reach $s_7$. The states $s_1,s_2,s_3$
have the same observation (for player~$1$), and the states $s_5,s_6$ as well.
The game starts in a \mbox{player-$2$} state $s_0$ with successors $s_1$ and $s_2$.
Thus after one step, player~$1$ does not know whether the game is in $s_1$
or in $s_2$. 
If the game is in $s_2$, then the action $b$ leads to $s_4$ from where
the target state $s_7$ is not reachable. Hence playing $b$ is not a good choice.
Playing $a$ gives either a new observation ($s_6$ is reached), or
the same observation as in the previous step ($s_3$ is reached). 
At this point, a simple observation-based strategy can play $b$ and reach
$s_7$ for sure. However, an observation-based stutter-invariant (or collapsed-stutter-invariant)
strategy must keep playing $a$ in $s_3$ (since the observation did not change)
and reaches $s_5$ with same observation as in $s_6$. Now, an observation-based 
\emph{stutter-invariant} strategy wins by playing $a$ in $s_5$ and $b$ in $s_6$.
Indeed with observation-based stutter-invariant strategies, player~$1$
can distinguish whether the game is in $s_5$ or in $s_6$ (simply looking at the 
length of the play prefix in this case). However, with an observation-based 
\emph{collapsed-stutter-invariant} strategy, player~$1$ cannot win because
the play prefixes $\rho_1=s_0 \dots s_5$ and $\rho_2=s_0 \dots s_6$ have the same 
collapsed stuttering sequence of observations, thus forcing player~$1$ to choose
the same action $a$ or $b$, and action $a$ is losing for $\rho_2$ and action $b$ for 
$\rho_1$. 
%%% which both are losing from either of the two prefixes.
Hence in this game there is no observation-based collapsed-stutter-invariant winning strategy,
but there is an observation-based stutter-invariant winning strategy.
%, illustrating the difference between these two notions.
%To show that observation-based stutter-invariant is different from observation-based collapsed-stutter-invariant ?
%
%{\bf TO BE DONE.} See \figurename~\ref{fig:stuttering} (stochastic game) and \figurename~\ref{fig:components} (component library).
%Note that the observation-based stutter-invariant strategies of Cassez et al.~\cite{CDLLR07}
%correspond to our observation-based collapsed-stutter-invariant strategies.
\medskip

The previous example shows that collapsed-stutter-invariant strategies are different
from stutter-invariant, as well as from standard observation-based strategies. 
Our goal is to decide the existence of finite-memory almost-sure winning strategies
in partial-observation stochastic parity games. 
This problem has been studied for observation-based strategies and optimal complexity
result (EXPTIME-completeness) has been established in~\cite{CDNV14}.
We now present a polynomial-time reduction for deciding the existence of finite-memory 
almost-sure winning collapsed-stutter-invariant strategies to observation-based strategies.
%\mynote{K: Added above paragraph.}

\smallskip\noindent{\bf Reduction of collapsed-stutter-invariant problem to observation-based problem.}
There are two main ideas of the reduction. 
(1)~The first is that whenever player~$1$ plays an action $a$, the
action $a$ is stored in the state space as long as the observation of the state
remains the same. This allows to check that player~$1$ plays always the same 
action along a sequence of identical observations. This only captures the stutter-invariant
restriction, but not the collapsed-stutter-invariant restriction.
(2)~Second, whenever a transition is executed, player~$2$ is allowed to loop 
arbitrarily many times through the new state. 
This ensures that player~$1$ cannot rely on the number of times he sees an observation,
thus that player~$1$ is collapsed-stutter-invariant. However, it should be 
forbidden for player~$2$ to loop forever in a state, which can be ensured by
assigning priority $0$ to the loop (hence player~$1$ would win the parity 
objective if the loop is taken forever by player~$2$).
We now formally present the reduction.
%\mynote{K: Changed above paragraph with small changes.}

\begin{figure}[!tb]
  \begin{center}

\begin{picture}(0,0)(0,0)
%\put(0,0){\framebox(5,5){}}
\gasset{Nw=9,Nh=9,Nmr=4.5,rdist=1, loopdiam=6}   %loopCW=y
\node[Nframe=n, Nmarks=n](loop)(2.7,-66.1){}   % 1 line vertical = 5 ?
\drawloop[ELpos=35, ELside=r, loopangle=90, loopdiam=4, loopCW=n, AHlength=0, linewidth=0.09](loop){{\scriptsize $\sharp$}}
\end{picture}

\begin{tabular}{|l|l|l|}
\cline{2-3}
  \multicolumn{1}{c|}{} & \multicolumn{2}{c|}{ {\Large\strut} $\ \ s \xrightarrow{\ a\ } s'$  \Large\strut} \\   %(in $G$)
\cline{2-3}
  \multicolumn{1}{c|}{} & \multicolumn{1}{c|}{$\obs(s) = \obs(s')$\Large\strut} & \multicolumn{1}{c|}{$\obs(s) \neq \obs(s')$\Large\strut} \\
  \multicolumn{1}{c|}{} & \multicolumn{1}{c|}{(storing last action of player~$1$)} & \multicolumn{1}{c|}{(forgetting last action of player~$1$)} \\
\hline
                        & $\tuple{s,x} \xrightarrow{\ a\ } \bot$ \quad for $x \neq a$ and $x \neq 0$ \Large\strut 
                        & $\tuple{s,x} \xrightarrow{\ a\ } \bot$ \quad for $x \neq a$ and $x \neq 0$ \Large\strut \\
  $s \in S_1$           & $\tuple{s,a} \xrightarrow{\ a\ } \tuple{s',\ov{a}}$  & $\tuple{s,a} \xrightarrow{\ a\ } \tuple{s',\ov{0}}$ \Large\strut \\
                        & $\tuple{s,0} \xrightarrow{\ a\ } \tuple{s',\ov{a}}$  & $\tuple{s,0} \xrightarrow{\ a\ } \tuple{s',\ov{0}}$ \Large\strut \\
\arrayrulecolor{mygray}\cline{2-3}\arrayrulecolor{black}
                        & \multicolumn{2}{l|}{in $\tuple{s,x}$, if $x \in A_1$, then player~$1$ should play the stored action $x$;\Large\strut} \\
                        & \multicolumn{2}{l|}{\phantom{in $\tuple{s,x}$,} if $x=0$, no action is stored and player~$1$ can choose any action.}  \\
\hline
  $s \in S_2$           & $\tuple{s,x} \xrightarrow{\ a\ } \tuple{s',\ov{x}}$  & $\tuple{s,x} \xrightarrow{\ a\ } \tuple{s',\ov{0}}$ \Large\strut \\
\arrayrulecolor{mygray}\cline{1-3}\arrayrulecolor{black}
                        & \multicolumn{2}{c|}{\Large\strut} \\
  $s \in S_1 \cup S_2$  & \multicolumn{2}{c|}{$\tuple{s,\ov{x}} \xrightarrow{\ a\ } \tuple{s,x}$ \Large\strut} \\
\arrayrulecolor{mygray}\cline{1-3}\arrayrulecolor{black}
                        & \multicolumn{2}{l|}{player~$2$ can play all actions available in the original game, and\Large\strut} \\
                        & \multicolumn{2}{l|}{repeat arbitrarily many times the current observation.} \\
\hline
\end{tabular}
    \caption{Game transformation for the reduction of collapsed-stutter-invariant 
    problem to observation-based problem.\label{fig:reduction-collapsed}}
  \end{center}
\end{figure}

\smallskip\noindent{\bf The formal reduction.}
The reduction is illustrated in \figurename~\ref{fig:reduction-collapsed} and 
formally presented below. % \mynote{L: added figure.}
Given a partial-observation stochastic game $G = \tuple{S,S_1,S_2,A_1,A_2,\transg}$
with observation mapping $\obs: S \to \Obs$, we construct a game 
$G' = \tuple{S',S'_1,S'_2,A_1,A'_2,\transgp}$ as follows:

\begin{itemize}
\item $S' = S \times (A_1 \cup \ov{A}_1 \cup \{0,\ov{0}\}) \cup \{\bot\}$ where 
$\ov{A}_1 = \{\ov{a} \mid a \in A_1\}$, assuming that $0 \not\in A_1$.  %\mynote{L: added assumption}. 
The states $\tuple{s,0}$ are a copy of the state space of the original game, and in the states
$\tuple{s,a}$ with $s \in S_1$ and $a \in A_1$, player~$1$ is required to play action $a$; 
in the states $\tuple{s,\ov{0}}$ and $\tuple{s,\ov{a}}$, player~$2$ can stay for arbitrarily many steps.
The state $\bot$ is absorbing and losing for player~$1$.

\item $S'_1 = S_1 \times (A_1 \cup \{0\}) \cup \{\bot\}$.

\item $S'_2 = S' \setminus S'_1 = (S_2 \times (A_1 \cup \{0\})) \cup (S \times (\ov{A}_1 \cup \{\ov{0}\}))$.

\item $A'_2 = A_2 \cup \{\sharp\}$, assuming $\sharp \not\in A_2$. 

\item The probabilistic transition function $\transgp$ is defined as follows:
for all player-$1$ states $\tuple{s,x} \in S'_1$ and actions $a \in A_1$:
	\begin{itemize}
	\item if $x \in A_1 \setminus \{a\}$, then let $\transgp(\tuple{s,x},a))(\bot) = 1$, 
	that is player~$1$ loses the game if he does not play the stored action;
	\item if $x = a$ or $x=0$, then for all $s' \in S'$ let \\
		$\transgp(\tuple{s,x},a))(\tuple{s',\ov{a}}) = \transg(s,a)(s')$ if $\obs(s') = \obs(s)$,
%$\transgp(\tuple{s,0},a))(\tuple{s',x}) = 0$ for all $x \neq a$;
and let \\
		$\transgp(\tuple{s,x},a))(\tuple{s',\ov{0}}) = \transg(s,a)(s')$ if $\obs(s') \neq \obs(s)$;
%$\transgp(\tuple{s,0},a))(\tuple{s',x}) = 0$ for all $x \neq 1$;
thus we store the action $a$ as long as the state observation does not change;

	\item All other probabilities $\transgp(\tuple{s,x},a))(\cdot)$ are set to $0$, for example 
	$\transgp(\tuple{s,0},a))(\tuple{s',y}) = 0$ for all $y \neq \ov{a}$;

	\end{itemize}

\item[] and for all player-$2$ states $\tuple{s,x} \in S'_2$, and actions $a \in A_2$:
	\begin{itemize}
	\item if $x \in A_1 \cup \{0\}$, then for all $s' \in S'$ let \\	
	$\transgp(\tuple{s,x},a))(\tuple{s',\ov{x}}) = \transg(s,a)(s')$ if $\obs(s') = \obs(s)$, and let\\ 
	$\transgp(\tuple{s,x},a))(\tuple{s',\ov{0}}) = \transg(s,a)(s')$ if $\obs(s') \neq \obs(s)$; thus 
	all actions are available to player~$2$ as in the original game, and the stored action $x$ 
	of player~$1$ is maintained if the state observation does not change; 

	\item if $x = \ov{b}$ for some $b \in A_1 \cup \{0\}$, then let \\
	$\transgp(\tuple{s,\ov{b}},\sharp))(\tuple{s,\ov{b}}) = 1$, and \\
	$\transgp(\tuple{s,\ov{b}},a))(\tuple{s,b}) = 1$ if $a \neq \sharp$; thus player~$2$
	can decide to stay arbitrarily long in $\tuple{s,\ov{b}}$ before going back
	to $\tuple{s,b}$;

	\item All other probabilities $\transgp(\tuple{s,x},a))(\cdot)$ and $\transgp(\tuple{s,x},\sharp))(\cdot)$ are set to $0$.
	\end{itemize}

%\item observations

%\item index function
\end{itemize}

The observation mapping $\obs'$ is defined according to the first component 
of the state: $\obs'(\tuple{s,x}) = \obs(s)$.
Given an index function $\alpha$ for $G$, define the index function $\alpha'$ for 
$G'$ as follows: $\alpha'(\tuple{s,x}) = \alpha(s)$ and $\alpha'(\tuple{s,\ov{x}}) = 0$
for all $s \in S$ and $x \in A_1 \cup \{0\}$, and $\alpha'(\bot) = 1$.
Hence, the state $\bot$ is losing for player~$1$, and the player-$2$ states $\tuple{s,\ov{x}}$ are
winning for player~$1$ if player~$2$ stays there forever.

\begin{lemma}\label{lem:reduction-collapsed-stutter-invariant}
Given a partial-observation stochastic game $G$ with observation mapping $\obs$
and parity objective~$\Phi_{\alpha}$ defined by the index function $\alpha$, a
game $G'$ with observation mapping $\obs'$ and parity objective $\Phi_{\alpha'}$ defined by the index function $\alpha'$
can be constructed in polynomial time such that the following statements are equivalent:
\begin{itemize}
\item there exists a finite-memory almost-sure winning observation-based collapsed-stutter-invariant strategy $\straa$ for player~$1$ in $G$ from $s_0$ for the parity objective $\Phi_{\alpha}$;
\item there exists a finite-memory almost-sure winning observation-based strategy $\straa'$ for player~$1$ in $G'$ from $\tuple{s_0,\ov{0}}$ for the parity objective $\Phi_{\alpha'}$.
\end{itemize}
\end{lemma}

\smallskip\noindent{\bf Correctness argument.}
The correctness proof has two directions: first, given a finite-memory almost-sure 
winning collapsed-stutter-invariant strategy in $G$ to construct a similar witness
of observation-based strategy in $G'$, and vice versa (the second direction).
We present both directions below.
%\mynote{K: Added above paragraph. In above lemma added word almost-sure winning. In first 
%direction added word finite-memory, though not necessary. In second direction, some occurrences
%of stutter-invariant changed collapsed-stutter-invariant. In theorem, added word finite-memory.}

\smallskip\noindent{\em First direction.}
For the first direction, given a finite-memory almost-sure winning collapsed-stutter-invariant 
strategy $\straa$ for player~$1$ in $G$,
we construct a finite-memory observation-based strategy $\straa'$ for player~$1$ in $G'$ 
such that $\straa'$ is almost-sure winning in $G'$. 

To define $\straa'(\rho')$ for a play prefix $\rho'$ in $G'$ such that $\last(\rho') \in S'_1$, we construct a
play prefix $\rho$ in the original game $G$ and define $\straa'(\rho') = \straa(\rho)$.
We construct $\rho$ as $\mu(\rho')$ where $\mu$ is a mapping that first removes from $\rho'$ 
all states of the form $\tuple{s,\ov{x}}$ for $s \in S$ and $x \in A_1 \cup \{0\}$, 
and then projects all the other states $\tuple{t,\cdot}$
to their first component $t$. It follows that $\rho = \mu(\rho')$ is a play prefix in $G$ and
if the observation sequences of two prefixes $\rho'_1, \rho'_2$ in $G'$ are the 
same (i.e., $\obs'(\rho'_1) = \obs'(\rho'_2)$), then the collapsed stuttering
of $\mu(\rho'_1)$ and $\mu(\rho'_2)$ is also the same. It follows that 
the constructed strategy $\straa'$ is well defined,
and since $\straa$ is collapsed-stutter-invariant, $\straa'$ is observation-based.

We show that $\straa'$ is almost-sure winning for the parity objective 
$\Phi_{\alpha'}$ in $G'$. 
Consider an arbitrary strategy $\strab'$ for player~$2$ in $G'$.
We can assume without loss of generality that:
\begin{itemize}
\item $\strab'$ is pure since when strategy $\straa'$ is fixed in $G'$,
we get a $1$-player stochastic games, and pure strategies are sufficient 
in $1$-player stochastic games~\cite{CDGH10};

\item no play compatible with $\straa'$ and $\strab'$ 
gets stuck in a state of the form $\tuple{s,\ov{x}}$ for $s \in S$ 
and $x \in A_1 \cup \{0\}$ (i.e., no play 
loops forever through a self-loop on some state $\tuple{s,\ov{x}}$ after some prefix $\rho'$); 
this assumption is also without loss of generality because if $\strab'$ is a 
spoiling strategy (i.e., it ensures against $\straa'$ that the parity objective 
is satisfied with probability less than $1$) and a play gets stuck after some prefix $\rho'$,
then we can define a strategy $\strab''$ that plays arbitrarily after $\rho'$
but does not get stuck, i.e., never plays $\sharp$, since getting stuck forever implies
winning for player~1. 
It follows that $\strab''$ is also a spoiling strategy and never gets stuck.
\end{itemize}

From the strategy $\strab'$ we define a strategy $\strab$ for player~$2$ 
in $G$ (that basically mimics $\strab'$ ignoring the $\sharp$ actions).
It follows by induction that (up to the mapping $\mu$) the probability 
measure over plays in $G$ under strategies $\straa$ and $\strab$ coincides 
with the probability measure over plays in $G'$ under strategies 
$\straa'$ and $\strab'$. Since $\mu$ preserves the satisfaction
of the parity objective, it follows that with probability~$1$ the parity
objective is satisfied in $G'$ under $\straa'$ and $\strab'$,
and thus $\straa'$ is an almost-sure winning observation-based strategy
in~$G'$.

\smallskip\noindent{\em Second direction.}
For the second direction of the lemma, consider that there exists a finite-memory
almost-sure winning observation-based strategy $\straa'$ for player~$1$ in $G'$ 
from state $\tuple{s_0,\ov{0}}$ for the parity objective $\Phi_{\alpha'}$, 
and we show that there exists an observation-based collapsed-stutter-invariant
almost-sure winning observation-based strategy for player~$1$ in $G$ from $s_0$.

Let $\tuple{\mem', m'_0, \straa'_u, \straa'_n}$ be a transducer that  encodes $\straa'$ (thus $\mem'$ is finite).
We construct a transducer $\tuple{\mem, m_0, \straa_u, \straa_n}$ and show that
it encodes an observation-based collapsed-stutter-invariant strategy $\straa$ that is almost-sure 
winning in $G$ from $s_0$ for the parity objective $\Phi_{\alpha}$. Intuitively, given
a memory value $m' \in \mem'$ and a sequence of states with identical observation $o$
visited in the game $G'$, 
the memory will be updated (according to $\straa'_u$) and the actions played (according to $\straa'_n$)
may be different depending on the number of repetitions of the observation $o$. 
To construct a collapsed-stutter-invariant strategy, we update the memory 
to a value that occurs infinitely often in the sequence of memory updates obtained
when observing $o$ repeatedly. This ensures that, as long as the observation
does not change, the constructed strategy plays always the same action. Moreover,
this action could indeed be played by the original strategy (even after an arbitrarily
long sequence of identical observations).  %\mynote{L: added some intuition above.}
The transducer for $\straa$ is defined as follows:

\begin{itemize}
\item $\mem = \mem' \times \Obs$. A memory value $m=\tuple{m',o}$ corresponds to 
memory value $m'$ in the transducer for $\straa'$, and the current observation is $o$.
As long as the next observation is $o$, the memory value $m$ does not change
(there is a self-loop on $m$ for all inputs $s$ such that $\obs(s) = o$).

\item $m_0 = \tuple{m', \obs(s_0)}$ where $m'$ is such that $m' = \straa'_u(m'_0, s_0^n)$ for infinitely many $n$,
and $s_0^n = \underbrace{s_0s_0\dots s_0}_{n \text{ times}}$ is the $n$-fold repetition of $s_0$ 
(such $m'$ always exists since $\mem'$ is finite); note that in the definition
of $m'$ we can equivalently replace $s_0^n$ by $s^n$ for any $s$ such that $\obs(s) = \obs(s_0)$
since the strategy $\straa'$ is observation-based.

\item	For all $\tuple{m'_1, o_1} \in \mem$ and $s \in S$, if $\obs(s) = o_1$,
	then $\straa_u(\tuple{m'_1, o_1}, s) = \tuple{m'_1, o_1}$ (self-loop), and 
	if $\obs(s) \neq o_1$, then $\straa_u(\tuple{m'_1, o_1}, s) = \tuple{m'_2, \obs(s)}$ 
	where $m'_2 = \straa'_u(m'_1, s^n)$ for infinitely many $n$.

\item $\straa_n(\tuple{m', o}) = \straa'_n(m')$.
\end{itemize}

First, we show that the strategy $\straa$ is almost-sure winning in $G$ from $s_0$.
The proof of this claim is by contradiction: assume that $\straa$ is not
almost-sure winning. Then there exists a spoiling strategy $\strab$ for player~$2$ such that 
the parity objective $\Phi_{\alpha}$ is satisfied with probability less than~$1$.
From $\strab$, we define a strategy $\strab'$ for player~$2$ in $G'$ intuitively 
as follows: 
the strategy $\strab'$ mimics the strategy $\strab$ when the current state is
of the form $\tuple{s,x}$ where $s \in S_2$ and $x \in A_1 \cup \{0\}$, and 
ensures the following invariant: given any prefix $\rho'$ in $G'$, the memory value of $\straa'$ 
after $\rho'$ is $m'$ if and only if the memory value of $\straa$ 
after $\mu(\rho')$ is of the form $\tuple{m',\cdot}$, where $\mu$ is the mapping defined in the first 
direction of the proof. Player~$2$ can always ensure this invariant as follows:
given the memory value of $\straa$ is updated to $\tuple{m',o}$, repeat the
self-loop on action $\sharp$ sufficiently many times to let the memory value
of $\straa'$ be updated to $m'$, which is always possible since by definition
of the transducer for $\straa$, the value $m'$ is ``hit'' infinitely often
in the transducer for $\straa'$ when a state with observation $o$ is visited forever.
For example, the strategy $\strab'$ stays in the initial state $\tuple{s_0,\ov{0}}$, 
which is a player~$2$ state, for $n$ steps where $n$ is such that 
$m' = \straa'_u(m'_0, s_0^n)$ where $m'$ is such that $\tuple{m', \obs(s_0)}$ is 
the initial memory value of $\straa$. 
It follows from this definition of $\strab'$ that for all play prefixes $\rho'$ 
in $G'$ that are compatible with $\straa'$ and $\strab'$ in $G'$, the play 
prefix $\mu(\rho')$ is compatible with $\straa$ and $\strab$ and has the same
probability as $\rho'$. Therefore the respective probability measures in $G$ and in $G'$
coincide (up to the mapping $\mu$) and since for all infinite plays $\rho$ in $G$
and $\rho'$ in $G'$ such that $\rho=\mu(\rho')$, we have $\rho \in \Phi_{\alpha}$
if and only if $\rho \in \Phi_{\alpha'}$, it follows that $\strab'$ is a spoiling
strategy in $G$ (the parity objective $\Phi_{\alpha'}$ is satisfied with probability 
less than~$1$ under strategies $\straa'$ and $\strab'$). This contradicts
that $\straa'$ is almost-sure winning. Hence the original claim  
that the strategy $\straa$ is almost-sure winning in $G$ holds.

Second, we show that the strategy $\straa$ is observation-based collapsed-stutter-invariant.
We have by induction that $\straa_u(m_0, \rho) = \straa_u(m_0, \rho \cdot s)$
for all play prefixes $\rho$ and states $s$ such that $\obs(s) = \obs(\last(\rho))$, and it follows
that $\straa(\rho) = \straa(\rho')$ if $\collapse(\rho) = \collapse(\rho')$
which concludes the argument.

\smallskip
%By Lemma~\ref{lem:reduction-collapsed-stutter-invariant}, and since partial-observation
%stochastic parity games with finite-memory strategies can be solved in EXPTIME
%for almost-sure winning~\cite{CDNV14}, we get the following result.
Since solving partial-observation stochastic parity games with finite-memory 
observation-based strategies is EXPTIME-complete for almost-sure winning~\cite{CDNV14}, 
we get an EXPTIME upper bound for games with observation-based collapsed-stutter-invariant strategy
by the reduction in Lemma~\ref{lem:reduction-collapsed-stutter-invariant}. 
The same complexity results hold for the class of observation stutter-invariant strategies,
by removing all $\sharp$-labelled self-loops for player~$2$ in the reduction
for collapsed-stutter-invariant strategies.
%%%presented in Section~\ref{sec:posp-games}.
We note that an EXPTIME lower bound can be established for those problems by a 
converse reduction that introduces in every transition an intermediate dummy state
with a different observation, thus two consecutive observations are always different,
and the observation-based strategies are also collapsed-stutter-invariant. 

%\mynote{K: Remark about lower bound and optimality to discuss.}
%\mynote{L: added above remark and updated the theorem.}

\begin{theorem}\label{theo:collapsed-stutter-invariant-exptime}
The qualitative problem of deciding whether there exists a finite-memory 
almost-sure winning observation-based collapsed-stutter-invariant (or stutter-invariant) 
strategy in partial-observation stochastic games with parity objectives 
is EXPTIME-complete.
%there is an observation-based collapsed-stutter-invariant strategy $\straa$ such that $\val^G(\straa,\Phi_{\alpha_{G}})(s)=1$
\end{theorem}

\subsection{The Complexity of Qualitative Realizability} We now present the complexity
result for qualitative realizability for DPW specifications via a reduction 
to the problem of deciding the existence of finite-memory almost-sure winning 
collapsed-stutter-invariant strategies in partial-observation games.
The reduction formalizes the intuition described in Remark~\ref{rem:partial}.
%\mynote{K: Added above paragraph and reference to remark.}

\smallskip\noindent{\bf Reduction of synthesis for DPW specifications
to collapsed-stutter-invariant problem.}
The reduction is analogous to the upper-bound reduction presented
in Section~\ref{sec:complexity-results}. Given a library $\L$ of width $D$ and
a DPW, the game we construct is the product of the game $G_{\L}$ with 
the DPW. The states are of the form $\tuple{q,i,p}$ where $\tuple{q,i}$
is a state of $G_{\L}$ and $p$ is a state of the DPW. The third
component $p$ is updated according to the deterministic transition
function $\delta_P$ of the DPW. Thus the successors of $\tuple{q,i,p}$
are of the form $\tuple{\cdot,\cdot,p'}$ where $p' = \delta_P(p,L_i(q))$.
For example, the transitions for player-$2$ states $s=\tuple{q,i,p}$, 
and actions $\sigma \in A_2$ are defined by  
\[
\transg_{\L}(\tuple{q,i,p},\sigma)(\tuple{q',j,p'}) =
\begin{cases}
\delta_i(q,\sigma)(q') &  \text{if } i = j \text{ and } p' = \delta_P(p,L_i(q)) \\
0 & \text{otherwise}
\end{cases} 
\]

The index function in the game is defined according to the 
third component of the states and according to the index function $\alpha_P$ of the DPW,
thus $\alpha_G(\tuple{q,i,p})=\alpha_P(p)$.
The observation mapping is defined by $\obs(\tuple{q,i,p}) = i$ if 
$q \not\in F_i$, and $\obs(\tuple{q,i,p}) = \tuple{q,i}$ if $q \in F_i$. 
Thus the state of the DPW is not observable, and only the components name and
exit states are observable. Note that $\Obs = [k] \cup \bigcup_{i=0}^k (F_i \times \set{i})$
where $[k]=\set{0,1,2,\ldots,k}$ and the number of components is $k+1$.
The correctness argument is established using similar arguments as in Section~\ref{sec:complexity-results},
by showing that a composer for $\L$ can be mapped to a collapsed-stutter-invariant strategy in $G_{\L}$
that is almost-sure winning, and vice versa we can construct a composer for $\L$
from an almost-sure winning collapsed-stutter-invariant strategy in $G_{\L}$
by the inverse mapping.

This reduction and Theorem~\ref{theo:collapsed-stutter-invariant-exptime}
show that the realizability problem with DPW specifications can be solved in
EXPTIME, and an EXPTIME lower bound is known for this problem~\cite{AK14}.

\begin{theorem}
The qualitative realizability problem for controlflow composition with 
DPW specifications is EXPTIME-complete.
\end{theorem}

\subsection{Undecidability of the Quantitative Realizability}
In this section we establish undecidability of the quantitative
realizability problem by a reduction from the quantitative decision 
problem for probabilistic automata (which is undecidable).

\smallskip\noindent{\bf Probabilistic automata.}
A probabilistic automaton $\aut=\tuple{\Sigma, Q, q_0, \delta }$ is a 
probabilistic transducer without outputs and exit states, i.e., $\Sigma$ 
consists of the input letters, $Q$ is the finite state space with initial state
$q_0$, and $\delta:Q \times \Sigma \to \dist(Q)$ is the probabilistic 
transition function. 
Consider a probabilistic automaton $\aut$ with an index function $\alpha$ on 
$Q$. A word is an infinite sequence of letters from $\Sigma$, and a 
\emph{lasso-shaped} word $w=w_1 (w_2)^\omega$ consists of a finite word 
$w_1$ followed by an infinite repetition of a non-empty finite word $w_2$.
Given a probabilistic automaton $\aut$ with index function $\alpha$, 
the quantitative decision problem of whether there exists a lasso-shaped word 
that is accepted with probability at least $\eta$ is undecidable, for rational
$\eta \in (0,1)$ given as input~\cite{PazBook},
and the undecidability proof holds even for index function for reachability, 
B\"uchi, or coB\"uchi objectives.

\smallskip\noindent{\bf The key ideas of reduction.} 
The key ideas of the reduction are as follows.
We consider a component for each letter of the input alphabet, and each
component has a unique exit (i.e., $|D|=1$). 
Hence a composer represents a choice of word, and since the state space of a 
composer is finite, a composer represents a lasso-shaped word. 
Conversely, for every lasso-shaped word there is a composer.
In each component, in the starting state, there is a choice by the environment 
among the states of the probabilistic automaton, and given the choice of a state, the 
probabilistic transition is executed according to the current state and choice
of letter (represented by the choice component), and finally there is a 
transition to the unique exit state.
To ensure that the choice in each component really chooses the correct state
of the probabilistic automaton we use the DPW.
The DPW keeps track of the current state of the probabilistic automaton, and 
requires that the choice from the starting state of the component matches
the current state (otherwise it accepts immediately).

\smallskip\noindent{\bf The reduction.}
Consider a probabilistic automaton $\aut=\tuple{\Sigma, Q, q_0, \delta }$ with an
index function $\alpha$ on $Q$. 
Let $\Sigma=\set{0,1,2,\ldots,k}$. %%and $Q=\set{1,2,\ldots,n}$. 
We construct a library $\L$ with $k+1$ components $M_0, M_1, \ldots, M_k$ each with 
a unique exit as follows.
We have $M_i=\tuple{\Sigma_I,\Sigma_O, Q_i, q_0^{i}, \delta_i,F_i,L_i}$ where 
the components are as follows:
\begin{compactenum}
\item $\Sigma_I=Q$;  and $\Sigma_O=Q \cup \set{\$}$; 
\item $Q_i=\set{q_0^{i}} \cup (Q \times\set{1,2}) \cup \set{\ex_i}$;
\item $F_i=\set{\ex_i}$; 
\item $L_i(q_0^{i})=L_i(\ex_i)= \$$ and $L_i(\tuple{q,j})=q$ for $q \in Q$ and 
$j \in \set{1,2}$;
and 
\item (a)~$\delta_i(q_0^{i},\sigma)(\tuple{\sigma,1})=1$ 
(i.e., given the start state and an input letter $\sigma=q$ corresponding to a
state of $Q$, the next state is $\tuple{q,1}$); 
(b)~$\delta_i(\tuple{q,1},\sigma)(\tuple{q',2})= \delta(q,i)(q')$ 
(i.e., irrespective of the input choice $\sigma$, the second component changes
from $1$ to $2$, and the first component changes according to the transition
function $\delta$ of 
$\aut$ for the choice of input letter $i$); 
and (c)~$\delta_i(\tuple{q,2},\sigma)(\ex_i)=1$ (i.e., irrespective of choice of 
$\sigma$ the next state is the unique exit state).
\end{compactenum}

\smallskip\noindent{\bf The DPW.} We now describe the DPW along with the 
library of components.
The DPW has alphabet $Q\cup \set{\$}$ and state space 
$(Q\times\set{0,1,2}) \cup \set{\top}$.
%%The initial state is $\tuple{q_0,0}$, 
%%\mynote{L: should the initial state be $\tuple{q_0,2}$ ?} 
%%\mynote{K: it should have been $\tuple{q_0,3}$, which i omit, maybe we keep as it is.}
%%and 
The transition function $\delta_P$ is as follows:  
\begin{compactenum}
\item $\delta_P(\tuple{q,0},\sigma)=\tuple{q,1}$ 
if $\sigma=q$, else if $\sigma\neq q$, then $\delta((q,0),\sigma)=\top$ (i.e.,
in a state where the second component is~0 the automaton expects to read the 
same input as the first component, and if it reads so it changes the 
second component from~0 to~1, otherwise it goes to the $\top$ state). 
This step corresponds to reading the choice from the start state of a component.

\item $\delta_P(\tuple{q,1},\sigma)=\tuple{\sigma,2}$ 
(i.e., when the second component is~1 it updates the first component according
to the transition read, and the second component changes from~1 to~2).
This step corresponds to reading the transition in a component that mimics the 
transition of the probabilistic automaton. 

\item $\delta_P(\tuple{q,2},\sigma)=\tuple{q,0}$ 
(i.e., irrespective of the input, the first component remains the same, and 
the second component changes from~2 to~0).
This step corresponds to reading the transition to the exit state in a 
component.
\end{compactenum}
To be very precise, one also needs to add more states in the DPW for transition
from the exit state of a component to the start state of the next component 
(which is omitted for simplicity).
The state $\top$ is an absorbing state (irrespective of the input the next 
state is $\top$ itself).
The index function $\alpha_P$ for the DPW maps according to the index function 
of $\alpha$ and the first component, i.e., 
for $\tuple{q,i}$ where $q \in Q$ and $i\in \set{0,1,2}$ we have 
$\alpha_P(\tuple{q,i})=\alpha(q)$; and $\alpha_P(\top)=0$.

\smallskip\noindent{\bf Correctness argument.}
Given the probabilistic automaton $\aut$ with index function $\alpha$, let
$\Phi_{\alpha}$ be the corresponding parity objective.
For a lasso-shaped word $w$, let $\Prb^w(\Phi_{\alpha})$ denote the probability 
that the word satisfies the parity objective.
Given a composer $C_w$ that corresponds to a lasso-shaped word $w$, consider
a strategy of the environment that given the starting state of a component 
always chooses the state according to the first component of the current state
of the DPW. 
Then it follows that the probability distribution of $\aut$ is executed, and
hence we have $\Prb^w(\Phi_{\alpha}) \leq \val(\T_{C_w}, \Phi_{\alpha_P})$, where $\Phi_{\alpha_P}$ is
the parity objective induced by the DPW. 
Conversely, if the environment does not choose according to the current state
of the DPW, then the DPW immediately accepts.
It follows that  $\Prb^w(\Phi_{\alpha}) \geq \val(\T_{C_w}, \Phi_{\alpha_P})$, and thus 
we have  $\Prb^w(\Phi_{\alpha}) = \val(\T_{C_w}, \Phi_{\alpha_P})$.
Hence the answer to the quantitative realizability problem is YES iff the 
answer to the quantitative decision problem for $\aut$ is YES.
We have the following result.

\begin{theorem}
The quantitative realizability problem for controlflow composition with 
DPW specifications is undecidable.
\end{theorem}

%% file: intro.tex
\begin{abstract}
The synthesis problem asks for the automatic construction of a 
system from its specification. 
In the traditional setting, the system is ``constructed from scratch''
rather than composed from reusable components. 
However, this is rare in practice, and almost every non-trivial 
software system relies heavily on the use of libraries of reusable
components. 
Recently, Lustig and Vardi introduced \emph{dataflow} and
\emph{controlflow} synthesis from libraries of reusable
components. They proved that dataflow synthesis is undecidable,
while controlflow synthesis is decidable.
The problem of controlflow synthesis from libraries of 
\emph{probabilistic components} was considered by Nain, Lustig and Vardi, 
and was shown to be decidable for qualitative 
analysis (that asks that the specification be satisfied with probability~1).
Our main contributions for controlflow synthesis from probabilistic 
components are to establish better complexity bounds for the qualitative 
analysis problem, and to show that the more general quantitative problem 
is undecidable. 
For the qualitative analysis, we show that the problem (i)~is EXPTIME-complete 
when the specification is given as a deterministic parity word automaton, 
improving the previously known 2EXPTIME upper bound; and 
(ii)~belongs to UP $\cap$ coUP and is parity-games hard, when the 
specification is given directly as a parity condition on the components, 
improving the previously known EXPTIME upper bound.
\end{abstract}

\section{Introduction}

\noindent{\em Synthesis from existing components.} 
Reactive systems (hardware or software) are rarely built from scratch, but 
are mostly developed based on existing components. 
A component might be used in the design of multiple systems, e.g., 
function libraries, web APIs, and ASICs. 
The construction of systems from existing reusable components is an active 
research direction, with several important works, such as 
component-based construction~\cite{Sif05}, 
``interface-based design''~\cite{dAH01}, 
web-service orchestration~\cite{BCGLM03}. 
The synthesis problem asks for the automated construction of a system 
given a logical specification. 
For example, in LTL (linear-time temporal logic) synthesis, 
the specification is given in LTL and the reactive system to be constructed 
is a finite-state transducer~\cite{PR89a}.  
In the traditional LTL synthesis setting, the system is ``constructed from 
scratch'' rather than ``composed'' from existing components.  
Recently, Lustig and Vardi introduced the study of synthesis 
from reusable or existing components~\cite{LV09}.

\smallskip\noindent{\em The model and types of composition.} 
The precise mathematical model for the components and their composition 
is an important concern (and we refer the reader to~\cite{LV09,NLV14} for a 
detailed discussion).
As a basic model for a component, following~\cite{LV09}, 
we abstract away the precise details of the component and model a component 
as a {\em transducer}, i.e., a finite-state machine with outputs. 
Transducers constitute a canonical model for reactive components, abstracting 
away internal architecture and focusing on modeling input/output behavior.  
In~\cite{LV09}, two models of composition were studied, namely, 
\emph{dataflow} composition, where the output of one component becomes an 
input to another component, and {\em controlflow} composition, where at 
every point of time the control resides within a single component.
The synthesis problem for dataflow composition was shown to be 
undecidable, whereas the controlflow composition was shown to be 
decidable~\cite{LV09}.  

\smallskip\noindent{\em Synthesis for probabilistic components.}
While~\cite{LV09} considered synthesis for non-probabilistic components,
the study of synthesis for controlflow composition for probabilistic 
components was considered in~\cite{NLV14}.
Probabilistic components  are transducers with a probabilistic transition 
function, that corresponds to modeling systems where there is probabilistic 
uncertainty about the effect of input actions.
Thus the controlflow composition for probabilistic transducers aims at 
construction of reliable systems from unreliable components.  
There is a rich literature about verification and analysis of such systems, 
cf.~\cite{Var85,CY90,CY95,Var99,BK08,KNP11}, as well as about synthesis in 
the presence of probabilistic uncertainty~\cite{BGLBC04}.  

\smallskip\noindent{\em Qualitative and quantitative analysis.}
There are two probabilistic notions of correctness, namely, the 
\emph{qualitative} criterion that requires the satisfaction of 
the specification with probability~1, and the more general 
\emph{quantitative} criterion that requires the satisfaction of the 
specification with probability at least $\eta$, given $0 < \eta \leq 1$.

\smallskip\noindent{\em The synthesis questions and previous results.}
In the synthesis problem for controlflow composition, the input 
is a library $\L$ of probabilistic components, and we consider 
specifications given as parity conditions (that allow us to consider
all $\omega$-regular properties, which can express all commonly used 
specifications in verification).
The \emph{qualitative (resp., quantitative) realizability} and synthesis 
problems ask whether there exists a \emph{finite} system~$S$ built from 
the components in $\L$,  such that, regardless of the input provided by the 
external environment, the traces generated by the system~$S$ satisfy the 
specification with probability~1 (resp., probability at least $\eta$). 
Each component in the library can be instantiated an arbitrary number of times in
the construction and there is no a-priori bound on the size of the
system obtained. 
The way the specification is provided gives rise to two different problems:
(i)~\emph{embedded parity realizability}, where the specification is given in 
the form of a parity index on the states of the components; and 
(ii)~\emph{DPW realizability}, where the specification is given as a 
separate deterministic parity word automaton (DPW).
The results of~\cite{NLV14} established the decidability of the qualitative
realizability problem, namely, in EXPTIME for the embedded parity realizability 
problem and 2EXPTIME for the DPW realizability problem.
The exact complexity of the qualitative problem and the decidability and 
complexity of the quantitative problem were left open, which we study in this 
work.

\begin{table}[h]
\begin{center}
{\small
%\scalebox{0.90}{
\begin{tabular}{|l@{\ }c@{\;} |*{2}{c|c|}}
\hline
\large{\strut}           &  & \multicolumn{2}{|c|}{Qualitative} & \multicolumn{2}{|c|}{Quantitative} \\
\cline{3-6}
\large{\strut}    &
                        &  \;Our Results\;    & \;Previous Results\;        &  \;Our Results\;    & \;Previous Results\; \\
\hline
Embedded Parity   &   & UP $\cap$ coUP & EXPTIME & UP $\cap$ coUP & Open \\
(with exit control) & & (Parity-games hard) & & (Parity-games hard) & \\
\hline
Embedded Parity   &   & PTIME & Not considered & PTIME & Not considered \\
(unrestricted exit control) & &  & &  & \\
\hline
DPW Specifications & & EXPTIME-c & 2EXPTIME & Undecidable & Open \\
\hline
\end{tabular}
%}
}
\end{center}
\caption{Computational complexity of synthesis from probabilistic components.\label{tab:complexity}}
\end{table}

\smallskip\noindent{\em Our contributions.}
Our main contributions are as follows (summarized in 
Table~\ref{tab:complexity}).
\begin{compactenum}
\item We show that both the qualitative and quantitative realizability 
problems for embedded parity lie in UP $\cap$ coUP, and even the qualitative 
problem is at least parity-games hard (the parity-games problem also 
belongs to  UP $\cap$ coUP~\cite{Jur98}, and the existence of a polynomial-time 
algorithm is a major and long-standing open problem).
Moreover, we show that a special case of the quantitative embedded parity 
problem (namely, unrestricted exit control) can be solved in polynomial time:
a probabilistic component has a set of exits, and in general there is a 
constraint on the current exit and the next component to transfer the control,
and in the unrestricted exit control problem no such constraint is imposed.
%%% KRISH: THE ABOVE PART WILL BE REMOVED FOR SHORT VERSION.
%%% In the full version we also show that an important special case of the 
%%% quantitative embedded parity problem can be solved in polynomial time. 

\item We show that the qualitative realizability problem for DPW 
specifications is EXPTIME-complete (an exponential improvement over the 
previous 2EXPTIME result).
Finally, we show that the quantitative realizability problem  for DPW 
specifications is undecidable.
\end{compactenum}

\noindent{\em Technical contributions.} 
Our two main technical contributions are as follows.
First, for the realizability of embedded parity specifications, while the most 
natural interpretation of the problem is as a partial-observation stochastic 
game (as also considered in~\cite{NLV14}), we show that the problem can be 
reduced in polynomial time to a perfect-information stochastic game.
Second, for the realizability of DPW specifications, we consider 
partial-observation stochastic games where the strategies correspond to a 
correct composition that defines, given an exit state of a component, to 
which component the control should be transferred.
Since we aim at a finite-state system, we need to consider strategies 
with \emph{finite memory}, and since the control flow is deterministic, 
we need to consider \emph{pure} (non-randomized) strategies.
Moreover, since the composition must be independent of the internal 
executions of the components, we need to consider strategies with 
\emph{stuttering} invariance. 
%%to obtain a composition that is independent of the steps executed 
%%by the components.
For example, we consider stutter-invariant strategies that must play the 
same when observations are repeated, and collapsed stutter-invariant 
strategies that are stutter-invariant strategies but not allowed to observe 
the length of the repetitions. 
%%While the qualitative analysis of partial-observation games with a parity 
%%condition is in general undecidable, we show that considering the class of 
%%pure finite-memory strategies with stuttering invariance leads to decidability.
We present polynomial-time reductions for both stutter-invariant and 
collapsed stutter-invariant strategies to games with standard observation-based 
strategies. Our results establish optimal complexity results for qualitative 
analysis of partial-observation stochastic games with finite-memory 
stutter-invariant and collapsed stutter-invariant strategies, which are of 
independent interest. 
Finally, we present a polynomial reduction of the qualitative realizability 
for DPW specifications to partial-observation stochastic games with collapsed 
stutter-invariant strategies and obtain the EXPTIME-complete result.

%% file: figures/stuttering-vs-collapsed-stuttering.tex
\begin{picture}(100,55)(0,0)
%\put(0,0){\framebox(112,70){}}

\gasset{Nw=9,Nh=9,Nmr=4.5,rdist=1, loopdiam=6}   %loopCW=y

%\node[Nmarks=i, Nmr=0](q0)(10,20){$q_0$}
%\node[Nmarks=n](q1)(30,30){$q_1$}
%\rpnode[Nmarks=n](r1)(50,30)(4,3.5){}

%\node[Nframe=n](label)(0,70){$\ov{G}$}

\node[Nmarks=i, Nmr=0](s0)(5,25){$s_0$}
\node[Nmarks=n](s1)(27,37){$s_1$}
\node[Nmarks=n](s2)(27,13){$s_2$}
\node[Nmarks=n](s3)(49,37){$s_3$}
\node[Nmarks=n](s4)(49,13){$s_4$}
\node[Nmarks=n](s5)(71,37){$s_5$}
\node[Nmarks=n](s6)(71,13){$s_6$}
\node[Nmarks=r](s7)(93,37){$s_7$}

\drawedge[ELpos=50, ELside=r, curvedepth=0](s0,s2){$a$}  %{$-,a$}
\drawedge[ELpos=50, ELside=l, curvedepth=0](s0,s1){$b$}  %{$-,b$}

\drawedge[ELpos=50, ELside=l, curvedepth=0](s2,s4){$b$}  %{$b,-$}
\drawedge[ELpos=50, ELside=r, curvedepth=-8](s2,s6){$a$}  %{$a,-$}
\drawedge[ELpos=50, ELside=l, curvedepth=0](s1,s3){$a,b$}  %{$a,-$}
%\drawedge[ELpos=50, ELside=r, curvedepth=0](s2,s4){$b$}  %{$b,-$}
\drawedge[ELpos=50, ELside=l, curvedepth=0](s3,s5){$a$}  %{$a,-$}
\drawedge[ELpos=50, ELside=l, curvedepth=8](s3,s7){$b$}  %{$b,-$}
\drawedge[ELpos=46, ELside=r, curvedepth=0](s6,s4){$a$}  %{$a,-$}
\drawedge[ELpos=40, ELside=r, curvedepth=0](s6,s7){$b$}  %{$b,-$}
\drawedge[ELpos=25, ELside=l, curvedepth=0](s5,s4){$b$}  %{$b,-$}
\drawedge[ELpos=50, ELside=r, curvedepth=0](s5,s7){$a$}  %{$a,-$}

%\drawloop[ELpos=40, ELside=l, ELdist=-2, loopangle=90, loopdiam=4](s3){\begin{tabular}{l}$a,-$\\[-3pt] $b,-$\end{tabular}}
%\drawloop[ELside=l, ELdist= 0, loopangle=90, loopdiam=4](s7){\begin{tabular}{l}$a,-$\\[-3pt] $b,-$\end{tabular}}
\drawloop[ELpos=40, ELside=l, ELdist=0, loopangle=90, loopdiam=4](s4){$a,b$}
\drawloop[ELside=l, ELdist= 1, loopangle=90, loopdiam=4](s7){$a,b$}

%\node[Nmarks=n, dash={1.5}0, Nw=12, Nh=32, Nmr=3](box)(30,20){}
\drawline[AHnb=0, dash={1 1.5}{1.3}](21,7)(21,43)(55,43)(55,31)(33,31)(33,7)(21,7)
\drawline[AHnb=0, dash={1 1.5}{1.3}](64.5,7)(64.5,43)(77,43)(77,7)(64.5,7)

%\node[Nframe=n](label)(105,2){$\underbrace{\phantom{zzezezezeez}}_{\text{highest odd priority}}$}

%\drawline[AHnb=0, arcradius=2, dash={.2 2}0](75,40)(75,2)(135,2)(135,80)(75,80)(75,40)
%\node[Nframe=n](label)(140,15){\begin{tabular}{l}same observation as $\ov{s}_0$\\priority $2d$\end{tabular}}

%\drawedge[dash={1}0](n3bis,nkbis){$0,1$}

\end{picture}